\documentclass[
	12pt, % allowed font sizes : 10pt (default), 11pt and 12pt
	a4paper, % paper sizes: a4paper (default), a5paper
	]{article}

\usepackage{amssymb} % advanced math symbols
\usepackage{amsmath} % for equations,  math functions and operators
\usepackage{amsthm}
\usepackage{geometry} % adjust margins below with \geometry{options}
\usepackage{graphicx} % \includegraphics preferable as pdf or png
\usepackage{caption,subcaption} % captions for figures and tables
\usepackage{setspace} % adjust line spacing with \onehalfspacing, \doublespacing and \singlespacing
\usepackage{natbib} % advanced bibliography management
\usepackage{array} % visualizes matrixes 
\usepackage[colorlinks,linkcolor=blue!70!black,filecolor=blue!70!black,urlcolor=blue!70!black,citecolor=blue!70!black]{hyperref}  % set hyperlinks in the document and to the outside (webpages, mail-adrdresses)
\hypersetup{breaklinks=true}   % splits links across lines
\usepackage[nameinlink]{cleveref} %for clever reference
\usepackage{url}

\usepackage{enumitem} %for enumerate

\usepackage{lineno} % show line numbers in the left margin
\usepackage{bm,dsfont,sidecap,comment,sgame}

\usepackage{tikz} %tikz to draw
\usetikzlibrary{arrows,positioning,calc}
\usetikzlibrary{shapes.geometric}

\newtheorem{theorem}{Theorem}
\newtheorem{corollary}{Corollary}
\newtheorem{lemma}{Lemma}

\newtheorem{set-up}{Set-up}

\newtheorem*{theorem*}{Theorem}
\newtheorem*{corollary*}{Corollary}
\newtheorem*{lemma*}{Lemma}
\newtheorem*{observation*}{Observation}
\newtheorem*{proposition*}{Proposition}
\newtheorem*{claim*}{Claim}
\newtheorem*{fact*}{Fact}
\newtheorem*{assumption*}{Assumption}
\newtheorem*{assumptionA*}{Assumption A}
\newtheorem*{set-up*}{Set-up}

\theoremstyle{definition}
\newtheorem{definition}{Definition}
\newtheorem*{definition*}{Definition}

\newtheorem*{problem*}{Problem}

\newtheorem*{example*}{Example}

\usepackage{apptools}
\AtAppendix{\counterwithin{lemma}{section}}
\AtAppendix{\counterwithin{proposition}{section}}

\renewenvironment{proof}[1][\proofname]{{\noindent\textbf{#1.}}}{\qed \vspace{\topsep}}
\newcolumntype{L}[1]{>{\raggedright\let\newline\\arraybackslash\hspace{0pt}}m{#1}}
\newcolumntype{C}[1]{>{\centering\let\newline\\arraybackslash\hspace{0pt}}m{#1}}
\newcolumntype{R}[1]{>{\raggedleft\let\newline\\arraybackslash\hspace{0pt}}m{#1}}

\DeclareMathOperator*{\argmax}{arg\,max}

\newcommand{\boldv}{\boldsymbol{v}}

\newcommand{\boldsigma}{\boldsymbol{\sigma}}

\newcommand{\boldc}{\boldsymbol{c}}
\newcommand{\boldq}{\boldsymbol{q}}

\newcommand{\boldw}{\boldsymbol{w}}
\newcommand{\boldmu}{\boldsymbol{\mu}}

\newcommand{\st}{\quad \text{subject to} \quad}

\begin{document}

\begin{titlepage}
\singlespacing
\title{Artificial Bugs for Crowdsearch\thanks{This research was partially supported by the Zurich Information Security and Privacy Center (ZISC). We thank Kari Kostiainen, Ueli Maurer, and workshop participants at ETH Zurich for valuable comments and discussions. All errors are our own.}}

% add the authors. Multiple authors are connected with "\and"
% respect the lexiographic order!
\author{
	Hans Gersbach\\
	\normalsize KOF Swiss Economic Institute, \\
        \normalsize ETH Zurich, and CEPR\\ 
	\normalsize Leonhardstrasse 21\\
	\normalsize 8092 Zurich, Switzerland\\
	\normalsize \href{mailto:hgersbach@ethz.ch}{hgersbach@ethz.ch}
	\and
	Fikri Pitsuwan\\
	\normalsize KOF Swiss Economic Institute, \\
        \normalsize ETH Zurich\\ 
	\normalsize Leonhardstrasse 21\\
	\normalsize 8092 Zurich, Switzerland\\
	\normalsize \href{mailto:fpitsuwan@ethz.ch}{fpitsuwan@ethz.ch}
        \and
        Pio Blieske\\
        \normalsize KOF Swiss Economic Institute, \\
        \normalsize ETH Zurich\\ 
	\normalsize Leonhardstrasse 21\\
	\normalsize 8092 Zurich, Switzerland\\
	\normalsize \href{mailto:pblieske@ethz.ch}{pblieske@ethz.ch}
	}

\date{Last updated: \today}

\maketitle
% \vspace{-0.8cm}
\begin{abstract}
    \noindent Bug bounty programs, where external agents are invited to search and report vulnerabilities (bugs) in exchange for rewards (bounty), have become a major tool for companies to improve their systems. We suggest augmenting such programs by inserting artificial bugs to increase the incentives to search for real (organic) bugs. Using a model of crowdsearch, we identify the efficiency gains by artificial bugs, and we show that for this, it is sufficient to insert only one artificial bug.  Artificial bugs are particularly beneficial, for instance, if the designer places high valuations on finding organic bugs or if the budget for bounty is not sufficiently high. We discuss how to implement artificial bugs and outline their further benefits.
	\\
	% !!! no additional linespaces here !!!
	\vspace{0in}
	\vspace{0in}
	\bigskip
\end{abstract}
\thispagestyle{empty}
\end{titlepage}

\pagebreak \newpage

\setcounter{page}{2}

%% Line numbering options
%\linenumbers

%% Linespacing options
%\singlespacing
%\onehalfspacing
\doublespacing

%%%%%%%%%%%%%%%%%%%%%%%%%%%%%%%%%%%%%%%%%%%%%%%%%%%%%%%%%%%%%%%%%%
%%%%%%%%%%%%%%%%%%%%%%%%%%%%%%%%%%%%%%%%%%%%%%%%%%%%%%%%%%%%%%%%%%
%%%%%%%%%%%%%%%%%%%%%%%%%%%%%%%%%%%%%%%%%%%%%%%%%%%%%%%%%%%%%%%%%%

\section{Introduction}

% Softwares and blockchains often have security vulnerabilities and can be attacked by adversaries, with potentially significant negative social or economic consequences. One such example occurred in 2019 when a ``significant flaw'' in the intended Swiss new e-voting system was discovered. With the danger of potential vote manipulation, the Federal Council of Switzerland paused the development and ordered a redesign of the system \citep{federalchancellery2019}. The attack discovery was part of 

How should we design a public intrusion test where external \textit{security researchers} were allowed to probe the software and report any vulnerabilities (bug) in exchange for rewards (bounty)? This type of program,  often called \textit{bug bounty} or \textit{crowdsourced security}, has become a major tool for detecting vulnerabilities in software used by governments, tech companies, and blockchains.\footnote{Their success in the past few years has led the Swiss authorities to systematically adopt bug bounty programs as a main measure in government cybersecurity. In a recent press release, the \citet{federaldepartmentoffinance2022} states that ``standardised security tests are no longer sufficient to uncover hidden loopholes. Therefore, in the future, it is intended that ethical hackers will search through the Federal Administration's productive IT systems and applications for vulnerabilities as part of so-called bug bounty programmes.''} Bug bounty is used by companies even if they have strong internal security teams (e.g. Google and Meta). It is also critical for blockchain infrastructure providers since such projects do not have dedicated security teams testing software upgrades. Once the software is deployed, there is no turning back and no legal mechanism defending against system exploitation, at least until the next hard fork\textemdash a major change in the blockchain protocol \citep{breidenbach2018,bohme2020}.

There have been comprehensive accounts on the rules of engagement of bug bounty programs \citep{laszka2018}, on the effectiveness and best practices of such programs \citep{walshe2020,malladi2020}, and on the incentives of researchers to participate in bug bounty programs~\citep{maillart2017}. For extensive literature on bug bounty programs, see \citet{bohme2006, zrahia2022,akgul2023}. 

In this paper, we suggest augmenting such programs by inserting artificial bugs to increase the incentives to search for real (organic) bugs. We start with a simple model of \textit{crowdsearch} (see, e.g., \citet{gersbach2023}), wherein researchers with different abilities decide on whether or not to exert costly effort to search for objects that are valuable to the contest designer. We then examine how inserting artificial bugs can be useful to the organization by lowering its financial commitment. We show that the ability to adjust the complexities of artificial bugs allows the organization to motivate more participants to search for organic bugs. We show that it is sufficient to insert one artificial bug to reap all possible efficiency gains. Moreover, we identify that an artificial bug is particularly beneficial if the designer has high valuations for finding organic bugs, if organic bugs are likely to exist, or if the designer's budget is low. Finally, we outline different engineering approaches to implement artificial bugs in practice and identify some other benefits of artificial bugs.

The paper is organized as follows: \Cref{sec:model} sets up the model. \Cref{sec:private} analyzes a private bug bounty program in which a finite number of agents search for bugs. \Cref{sec:public} considers a public program with a large crowd. Numerical examples are shown in \Cref{sec:example}. \Cref{sec:implement} outlines three approaches to implement artificial bugs and \Cref{sec:discuss} discusses other benefits of artificial bugs. \Cref{sec:conclusion} concludes. Proofs are relegated to the appendix.

%%%%%%%%%%%%%%%%%%%%%%%%%%%%%%%%%%%%%%%%%%%%%%%%%%%%%%%%%%%%%%%%%%%%%%
%%%%%%%%%%%%%%%%%%%%%%%%%%%%%%%%%%%%%%%%%%%%%%%%%%%%%%%%%%%%%%%%%%%%%%

\section{The Model} \label{sec:model}

There are potentially $L$ organic bugs in a system, indexed by $l$, and each exists with probability $\boldmu = (\mu^1,\dots,\mu^L) \in (0,1]^L$. We assume that the occurrence of bugs $l$ and bug $l'$ ($l \neq l'$) are stochastically independent, and neither the organization nor potential participants in a bug bounty program have any knowledge of such bugs. An organization (henceforth, \textit{designer}) values finding the bugs at $\boldw = (w^1,\dots,w^L) \in [0,\infty)^L$. The values of $w$ can be interpreted as utility or monetary values. The designer invites a set of risk-neutral agents $N$ to search for them in exchange for rewards $\boldv = (v^1,\dots,v^L) \in [0,\infty)^L$. If agent $i$ decides to search, s/he finds bug $l$ with probability $q^l \in (0,1]$, which represents its complexity, and receives a prize $v^l$ uniformly randomly rewarded to one of the agents who found it.\footnote{The assumption that all agents find a particular bug with the same probability is only for ease of exposition. The model can straightforwardly be extended to allow for heterogeneous bug-finding probability, with $q_i^l$ denoting agent $i$'s probability of finding bug $l$. This heterogeneity captures the idea that each security researcher has a particular skill set that may favor finding some bugs rather than others. \citet{maillart2017}'s empirical investigation of bug bounty programs supports this claim.} The events finding bug $l$ and finding bug $l'$ ($l \neq l'$) are stochastically independent. Searching, however, entails a fixed cost $c_i$ to agent $i \in N$, which is private information and drawn from a distribution $F$ with support on $[\underline{c},\overline{c}]$, where $-\infty \leq \underline{c} < \overline{c} \leq \infty$ and $\overline{c} > 0$.\footnote{The cost includes time spent reading and understanding the update, writing a program, designing an attack environment, and the opportunity cost of not participating in other bug bounty programs or other activities. We allow the possibility that the lower support $\underline{c}$ is negative to capture agents with intrinsic gain from search.} We impose that $F$ is continuous, has full support, has a finite density function $f$, and that $F/f$ is non-decreasing.\footnote{This holds for many common distributions such as the exponential distribution, the beta distribution with shape parameter $\beta \geq 1$, and $F(c)=\left(\frac{c-\underline{c}}{\overline{c}-\underline{c}} \right)^\alpha$ on $[\underline{c}, \overline{c}]$ with $\alpha > 0$, which includes the uniform distribution for $\alpha = 1$.}

The designer has a budget of $\overline{v} > 0$ and can set prizes $\boldv$ to incentivize the agents. In addition to setting the prizes for the organic bugs, the designer can insert $K$ artificial bugs, indexed by $k$, into the system, offering rewards $\boldv_a = (v_a^1,\dots,v_a^K) \in [0,\infty)^K$ and setting their complexities $\boldq_a = (q_a^1,\dots,q_a^K) \in [0,1]^K$, where complexity $q_a^k$ is the probability that an agent finds the artificial bug. Artificial bugs have no inherent value to the designer but may nonetheless be useful for the designer in motivating agents to participate.

The game has two stages. In the first stage, the designer sets $\boldv$, $\boldv_a$, and $\boldq_a$. In the second stage, agents decide whether or not to exert costly effort to search for bugs and are rewarded accordingly if they find any. The designer receives their value for the bugs found and pays out the prizes accordingly. We analyze the game by first solving for the (Bayes-Nash) equilibrium behavior of the agents in the second stage, then for the designer's optimal prizes and complexities of the artificial bugs.

We consider two scenarios. First, \Cref{sec:private} considers a \textit{private program}, where the set of agents is finite, $N = \{1,\dots,n\}$. This represents a private bug bounty program, where the designer invites $n$ agents to search for bugs. Second, \Cref{sec:public} considers a \textit{public program}, where we perform an asymptotic analysis of the game in which the number of agents tends to infinity, $n \rightarrow \infty$. This captures a public bug bounty program, where the public can freely choose to participate.

%%%%%%%%%%%%%%%%%%%%%%%%%%%%%%%%%%%%%%%%%%%%%%%%%%%%%%%%%%%%%%%%%%%%%%
%%%%%%%%%%%%%%%%%%%%%%%%%%%%%%%%%%%%%%%%%%%%%%%%%%%%%%%%%%%%%%%%%%%%%% 

\section{Private Program} \label{sec:private}

This section considers a private program for bug bounty, where the designer invites $n$ agents to search for bugs. First, we characterize the equilibrium behavior of the agents in the second stage. Then, we solve for the designer's optimal choices and determine that it is sufficient to insert only one artificial bug. We end this section by identifying the condition under which inserting an artificial bug is useful.

\subsection{Equilibrium Characterization}

Given $\boldv$, $\boldv_a$, and $\boldq_a$ set by the designer, the agents choose whether or not to search for the bugs. Since participants can only take binary choices, it is without loss of generality to consider only threshold strategies of the agents. A threshold strategy stipulates that an agent searches if and only if his/her private cost is below a certain threshold. We also focus on the symmetric equilibrium, where all agents use the same threshold $c^*$. Such an equilibrium always exists.

To characterize the equilibrium threshold, denote by $\Phi(\hat{c};q)$ the probability that an agent wins the prize associated with a bug with complexity $q$, conditioning on it existing, when the other $n-1$ agents are using a threshold strategy $\hat{c}$.\footnote{This event occurs if the agent finds the bug and is chosen among those who also found it.} Then, agent $i$'s expected benefit of searching is 
\begin{equation*}
\Psi(\hat{c};\boldv,\boldv_a,\boldq_a) \equiv \sum_{l} v^l\mu^l \Phi(\hat{c};q^l) + \sum_{k} v_a^k \Phi(\hat{c};q_a^k),
\end{equation*}
where the first term sums up the expected benefit from finding organic bugs and the second term from finding artificial bugs. Now, suppose that the other $n-1$ agents use the equilibrium threshold $c^*$, then agent $i$ will search if and only if their expected benefit of searching exceeds their private cost: $c_i \leq \Psi(c^*)$. The following result follows.

\begin{lemma} \label{lemma:equilibrium}
    If $\Psi(\underline{c}) \leq \underline{c}$, then $c^* = \underline{c}$. If $\Psi(\overline{c}) \geq \overline{c}$, then $c^* = \overline{c}$. Otherwise, the symmetric equilibrium threshold is $c^* = c^*(\boldv,\boldv_a,\boldq_a)$ is the unique solution to 
    \begin{equation} \label{eq:equilibrium}
        \hat{c} = \Psi(\hat{c};\boldv,\boldv_a,\boldq_a).
    \end{equation}
\end{lemma}

The equilibrium threshold $c^*$ is the unique fixed point of $\Psi$ and, therefore, the comparative statics properties are derived from the properties of $\Psi$, which in turn follows the properties of $\Phi$. The function $\Phi(\hat{c};q)$ is strictly decreasing in $\hat{c}$ and strictly increasing in $q$. The comparative statics results are as follows. Increasing $v^l$, $\mu^l$, and $v_a^k$, for some $l$ or $k$ increases $c^*$, as agents are more incentivized to search if the prizes and the probabilities that bugs exist are high. Moreover, decreasing the complexities\textemdash increasing $q^l$'s\textemdash also increases $c^*$.

\subsection{Optimal Prizes and Number of Artificial Bugs}

We now consider the designer's problem of choosing the optimal prizes (for organic and artificial bugs) and complexities (for artificial bugs). Let the probability that a bug with complexity $q$ is found be denoted by 
\begin{equation*} 
P(\hat{c};q) \equiv 1 - (1-qF(\hat{c}))^n
\end{equation*}
when agents follow the threshold $\hat{c}$ to decide whether to search or not. The designer chooses $\boldv$, $\boldv_a$, and $\boldq_a$ to maximize their objective function consisting of two components. The first is $\sum_l w^l \mu^l P(c^*;q^l)$, which is the sum of the values gained from any organic bugs found. The second component captures the rewards paid out to the agents for any organic \textit{and} artificial bugs found: $-\sum_l v^l \mu^l P(c^*;q^l) - \sum_{k} v_a^k P(c^*;q_a^k)$.

Therefore, the designer with a prize budget of $\overline{v}$ solves
\begin{equation}\label{eq:designer_problem}
\begin{aligned}
    \max_{\boldv, \boldv_a, \boldq_a} \; \sum_{l} (w^l-v^l) \mu^l P(c^*;q^l) - \sum_{k} v_a^k P(c^*;q_a^k) \\ \st \sum_{l} v^l + \sum_k v_a^k \leq \overline{v},\; v^l \geq 0,\; v_a^k \geq 0,\; q_a^k \in [0,1],
    \end{aligned}
\end{equation}
where $c^* = c^*(\boldv,\boldv_a,\boldq_a)$ is the symmetric equilibrium threshold that uniquely solves Equation \eqref{eq:equilibrium}. Let $\boldv^*$, $\boldv_a^*$, and $\boldq_a^*$ denote a solution to the designer's problem. The next result states that it is without loss of generality to insert only one artificial bug, if any.

\begin{lemma} \label{lemma:onebug}
If $(\boldv^*,\boldv_a^*,\boldq_a^*)$ is a solution, then there exists another solution $(\boldv^{**},\boldv_a^{**},\boldq_a^{**})$ such that $\boldv_a^{**} = (v_a^{1**},0,\dots,0)$.
\end{lemma}

\Cref{lemma:onebug} holds because it is inefficient to have multiple artificial bugs of different complexities, since artificial bugs are only useful indirectly, namely by incentivizing the agents to participate. For this purpose, it is always most efficient when the budget is spent on the prize corresponding to the artificial bug with the lowest complexity. In light of this, the designer's problem boils down to choosing $\boldv$, $v_a$, and $q_a$.

The key to solving the designer's problem is to note that $P(\hat{c};q)$ can be decomposed in terms of each agent's probability of winning the corresponding prize.

\begin{lemma} \label{lemma:keyrelation}
For a bug with complexity $q$, we have $P(\hat{c};q) = nF(\hat{c}) \Phi(\hat{c};q).$
\end{lemma}

Intuitively, if a bug is found, one of the agents has to win the prize. To win the prize, the agent has to have a cost below $\hat{c}$ \textit{and} to find the bug \textit{and} be chosen among all agents that search and also find the bug. Ex-ante, the former happens with probability $F(\hat{c})$ and the latter with probability $\Phi(\hat{c};q)$. From the perspective of the designer, all individuals are ex-ante identical; thus, multiplying by $n$ yields the probability that the bug is found.\footnote{The relation continues to hold if it is common knowledge that agents are heterogeneous in other aspects. For instance, agents may have different probabilities of finding a certain bug, or their costs may be drawn from different distributions. In such cases, the equilibrium thresholds differ across agents and we have the relation: $P((\hat{c}_1,\dots,\hat{c}_n);(q_1,\dots,q_n)) = \sum_{i} F_i(\hat{c}_i) \Phi_i(\hat{c}_{-i};(q_i,q_{-i}))$.} 

With \Cref{lemma:keyrelation}, the designer's objective function simplifies to $\sum_{l} w^l\mu^l P(\hat{c};q^l) - n F(\hat{c})\hat{c}$. Therefore, the designer's problem is to induce an optimal equilibrium threshold, denoted $\hat{c}^*$, that maximizes the objective function, subject to the equilibrium threshold being \textit{achievable} given the budget $\overline{v}$.\footnote{Achievability means that the available budget is sufficient to pay all prizes that have been promised.} Define the set achievable equilibrium thresholds as
\begin{equation*}
    C(\overline{v}) \equiv \left\{ \hat{c}: \hat{c} = \Psi(\hat{c};\boldv,v_a,q_a), \; \sum_{l} v^l + v_a \leq \overline{v},\; v^l \geq 0,\; v_a \geq 0,\; q_a \in [0,1] \right\}.
\end{equation*}
The set of achievable thresholds is an interval. Given that $\Phi(\hat{c};q)$ is strictly increasing in $q$, the highest achievable threshold is when the entire budget goes to the artificial bug that is found with certainty.

\begin{lemma} \label{lemma:achievablec}
    Let $c_a(\overline{v})$ be the unique fixed point of $\overline{v}\Phi(\hat{c};1)$. Then, $c_a(\overline{v})$ increases in $\overline{v}$ and $C(\overline{v}) = [0,c_a(\overline{v})]$.
\end{lemma}

Taken together, the previous lemmata show that in effect, the designer's problem is to maximize
\begin{equation} \label{eq:designerobj}
   W(\hat{c}) \equiv \sum_{l} w^l\mu^l P(\hat{c};q^l) - n F(\hat{c})\hat{c},
\end{equation}
subject to $\hat{c} \in [0,c_a(\overline{v})]$.

Our first main result characterizes the solution to the designer's problem. Define 
\begin{equation}\label{eq:foc}
    \Omega(\hat{c}) \equiv \sum_l w^l \mu^l q^l (1-q^lF(\hat{c}))^{n-1} - \frac{F(\hat{c})}{f(\hat{c})}
\end{equation}
and let $\tilde{c}$ be its unique fixed point, which is increasing in $w^l$ and $\mu^l$. We have

\begin{theorem} \label{thm:opt}
The optimal equilibrium threshold is $\hat{c}^* = \min\{\tilde{c},c_a(\overline{v})\}$. A set of prizes $\boldv^*$, $v_a^*$ and complexity $q_a^*$ that solve
$\hat{c}^* = \Psi(\hat{c}^*;\boldv, v_a,q_a)$ is optimal.
\end{theorem}

\subsection{Usefulness of Artificial Bugs}

Next, we identify when inserting an artificial bug is useful. Note that it is feasible for the designer not to add any artificial bug into the system, as $v_a = 0$ or $q_a = 0$ are feasible. In other words, the set of achievable equilibrium thresholds without any artificial bug is defined as
\begin{equation*}
    C_0(\overline{v}) \equiv \left\{ \hat{c}: \hat{c} = \Psi(\hat{c};\boldv,0,0), \; \sum_{l} v^l \leq \overline{v},\; v^l \geq 0 \right\} \subseteq  C(\overline{v}).
\end{equation*}
Analogously to \Cref{lemma:achievablec}, let $c^l(\overline{v})$ be the unique fixed point of $\overline{v}\mu^l\Phi(\hat{c};q^l)$. Then, $C_0(\overline{v}) = [0,c_0(\overline{v})]$, where $c_0(\overline{v}) = \max_{l} c^l(\overline{v})$. This leads to the next main result.

\begin{theorem} \label{thm:artificialbug}
    Inserting an artificial bug is beneficial to the designer if and only if $\tilde{c} > c_0(\overline{v})$.
\end{theorem}

The condition in \Cref{thm:artificialbug} holds if $w^l$ and $\mu^l$ are high, or $\overline{v}$ is not too high. These cases will be illustrated in the numerical examples in \Cref{sec:example}. In other words, inserting an artificial bug is beneficial if the designer has high valuations for finding organic bugs, if organic bugs are likely to exist, or if the designer's budget is low.

%%%%%%%%%%%%%%%%%%%%%%%%%%%%%%%%%%%%%%%%%%%%%%%%%%%%%%%%%%%%%%%%%%%%%%
%%%%%%%%%%%%%%%%%%%%%%%%%%%%%%%%%%%%%%%%%%%%%%%%%%%%%%%%%%%%%%%%%%%%%% 

\section{Public Program}  \label{sec:public}

 In this section, we examine the public bug bounty program and, in particular, the asymptotic behavior of the game as $n \rightarrow \infty$. This is of special interest since public bug bounty programs make use of the knowledge of a large group of experts. Throughout the section, we denote the equilibrium threshold and the equilibrium success probability when there are $n$ agents in the public program by $c_n = c^*(n)$ and $P_n(q) = P(c^*(n);q)$, respectively. We assume throughout the section that $\underline{c} > 0$. This is reasonable in most circumstances since even high-ability agents have to exert effort to find the bugs and communicate with the designer.\footnote{For $\underline{c}=0$, it holds that $P_n(q) \rightarrow 1$ for any prizes and complexities (see \citet{gersbach2023}). This case seems to be less plausible.} For ease of exposition, we impose two further assumptions, which exclude pathological cases in which a bug bounty program can not generate any benefits. First, we assume that $\overline{v} \Psi(\underline{c})\geq \underline{c}$, i.e., that the budget constraint is greater than the smallest possible cost for an agent to search. Otherwise, no agent will participate, and there is no point in designing a bug bounty program. Second, we assume that $\sum_l w^l \mu^l q^l \geq \underline{c}$ since otherwise, the designer can never compensate an agent for searching. Again, the budget constraint $\overline{v}$ and the bugs described by $\boldw$ and $\boldq$ are given for the designer.

We proceed as follows: First, we obtain results for the equilibrium behavior of the agents in the second stage as $n \rightarrow \infty$. Then, given the equilibrium behavior of the agents, we solve for the designer's optimal choices. Next, we verify that this is equivalent to solving first for the designer's optimal choices (i.e., the full game) and then performing an asymptotic analysis on such choices. Lastly, we analyze the usefulness of inserting an artificial bug in the public program.

\subsection{Asymptotic Behavior}

Given $\boldv$, $v_q$, and $q_a$, the following lemma summarizes the asymptotic results for the second stage game:

\begin{lemma}\label{lemma:limits}
    Let $\kappa^*$ be the unique fixed point of 
    \begin{equation*}
    \Psi_{\infty}(\hat{\kappa};\boldv,v_a,q_a) \equiv \sum_l v^l \mu^l \frac{1-e^{-q^l \hat{\kappa}}}{\underline{c}} + v_a \frac{1-e^{-q_a \hat{\kappa}}}{\underline{c}}.
    \end{equation*} 
    Then, for $n \rightarrow \infty$
    \begin{enumerate}
        \item[(i)] $c_n \rightarrow \underline{c}$,
        \item[(ii)] $n F(c_n) \rightarrow \kappa^*$,
        \item[(iii)] $P_n(q) \rightarrow P_\infty(q)=1-e^{-q \kappa^*}$.
    \end{enumerate}
\end{lemma}

Note that the asymptotic behavior is distribution-free insofar that it does not depend on the distribution of $F$, but only on $\underline{c}$. The reason is that only agents with the highest ability (lowest cost) participate in a public program. The quantity $\kappa^*$ is important as it represents the \textit{asymptotic participation}, i.e., the asymptotic expected number of participants in the crowdsearch. 

\subsection{Optimal Prizes}

Denote by $W_n(\hat{c})$ the designer's objective function \eqref{eq:designerobj} for a particular value of $n$. With \Cref{lemma:limits}, we obtain the designer's objective function at the equilibrium threshold, $c_n$, as $n \rightarrow \infty$.

\begin{corollary}\label{cor:public_utility}
    The designer's objective function at the equilibrium threshold $c_n$ converges:
    \begin{equation*}
        W_n(c_n) \rightarrow W_\infty(\kappa^*) \equiv \sum_l w^l \mu^l  \left(1- e^{-q^l \kappa^*} \right) - \kappa^* \underline{c}, \quad \text{ as } n \rightarrow \infty.
    \end{equation*}
\end{corollary}

Thus, the designer's payoff only depends on the asymptotic participation. As in the private program, we define the set of achievable $\kappa^*$'s given a budget constraint $\overline{v}$ as 
\begin{equation*}
    K(\overline{v}) \equiv \left \{ \hat{\kappa}:  \hat{\kappa} = \Psi_\infty(\hat{\kappa}; \boldv, v_a, q_a), \; \sum_l v^l + v_a \leq \overline{v}, \; v^l \geq 0, \; v_a \geq 0, \; q_a \in [0,1] \right\}.
\end{equation*}
The following lemma characterizes this set.

\begin{lemma}\label{lemma:achievable}
    Denote by $\kappa_a(\overline{v})$ the unique fixed point of $\overline{v}\frac{1-e^{-\hat{\kappa}}}{\underline{c}}$. Then, $\kappa_a(\overline{v})$ is increasing in $\overline{v}$ and $K(\overline{v})=[0, \kappa_a(\overline{v})]$.
\end{lemma}
Taking the previous lemmata together, the designer's problem becomes
\begin{equation*}
    \max_{\hat{\kappa} \in [0, \kappa_a(\overline{v})]} \sum_l w^l \mu^l  \left(1- e^{-q^l \hat{\kappa}} \right)- \hat{\kappa} \underline{c}.
\end{equation*}

Analogously to the private program, define
\begin{equation*}
    \Omega_{\infty}(\hat{\kappa}) \equiv \sum_l w^l \mu^l q^l e^{-q^l \hat{\kappa}} - \underline{c}.
\end{equation*}
We observe that $\Omega_{\infty}(\hat{\kappa})$ is strictly decreasing and thus has a unique root since $\sum_l w^l \mu^l q^l \geq \underline{c}$. Let $\tilde{\kappa}$ be its unique root, which is decreasing in $w^l$ and $\mu^l$. We obtain the following characterization.

\begin{theorem}\label{thm:public_opt}
    The optimal asymptotic participation is $\hat{\kappa}^*=\min\{\tilde{\kappa}, \kappa_a(\overline{v})\}$ and any set of prizes $\boldv^\infty$, $v_a^\infty$ and complexity $q_a^\infty$ that solve $\hat{\kappa}^* = \Psi_\infty(\hat{\kappa}^*; \boldv, v_a, q_a)$ is optimal.
\end{theorem}

\subsection{Verification}

We now verify that the designer's optimal prizes and complexities in a game with finite agents indeed converge to the solution obtained in \Cref{thm:public_opt}. Denote the set of optimal prizes and complexities in a private program with $n$ agents by 
\begin{equation*}
M_n = \{(\boldv, v_a, q_a) : \hat{c}^* = \Psi(\hat{c}^*;\boldv,v_a,q_a) \}
\end{equation*} 
and the set of optimal prizes and complexities in a public program by 
\begin{equation*}
M_\infty = \{(\boldv, v_a, q_a) : \hat{\kappa}^* = \Psi_\infty(\hat{\kappa}^*;\boldv,v_a,q_a) \}.
\end{equation*}
To show that the two sets asymptotically coincide, we introduce the Hausdorff distance which intuitively describes the maximum distance from points in any two sets.

\begin{definition}\label{def:hausdorff}
    Let $A, B \subseteq \mathbb{R}^k$ be two non-empty subsets. The Hausdorff distance induced by the Euclidean norm $\| \cdot\|_2$ is given by
    \begin{equation*}
        d(A, B)=\max \left \{ \sup_{x \in A} \inf_{y \in B} \|x - y\|_2 , \sup_{x \in B} \inf_{y \in A} \|x - y\|_2 \right \}.
    \end{equation*}
\end{definition}

Equipped with this, we state our next theorem.

\begin{theorem}\label{thm:limit_set}
    It holds that $d(M_n, M_\infty) \rightarrow 0$ as $n \rightarrow \infty$.
\end{theorem}

Hence, as $n$ increases, the maximal distance for any solution of the private program comes arbitrarily close to a solution of the public program.

\subsection{Usefulness of Artificial Bugs}

Analogously to \Cref{thm:artificialbug}, we identify when inserting an artificial bug is useful in a public program. Define $\kappa_0(\overline{v})=\max_l \kappa_l(\overline{v})$, where $\kappa_l(\overline{v})$ the fixed point of $v^l \mu^l \frac{1-e^{-q^l \kappa}}{\underline{c}}$.

\begin{theorem}\label{thm:public_artificial_bug}
    In the public program, inserting an artificial bug is beneficial to the designer if and only if $\tilde{\kappa} > \kappa_0(\overline{v})$. This is the case if and only if there exists $N\in \mathbb{N}$ such that for all $n \geq N$, it is beneficial to insert an artificial bug in the private program with $n$ agents.
\end{theorem}

Before ending this section, a remark on the cost distribution is in order. Throughout this section, we kept the distribution $F$ of the cost of the agents fixed. However, the distribution can plausibly vary, depending on the size of the crowd, $F_n$. Because with increasing $n$, less selection of the agents can occur until the whole population is open to participating, it is safe to assume that $F_n$ converges to some distribution. By \Cref{lemma:limits}, the limit case does not depend on the distribution, and all the results still hold, where the distribution $F_n$ converges to some limit distribution $F_\infty$.

%%%%%%%%%%%%%%%%%%%%%%%%%%%%%%%%%%%%%%%%%%%%%%%%%%%%%%%%%%%%%%%%%%%%%%
%%%%%%%%%%%%%%%%%%%%%%%%%%%%%%%%%%%%%%%%%%%%%%%%%%%%%%%%%%%%%%%%%%%%%%

\section{Numerical Examples} \label{sec:example}

We next illustrate the results with numerical examples and visualize the effects of an artificial bug and the dependence on the budget constraint.

\subsection{Private Program}

We start with the private program. Consider the uniform distribution on $[0,1]$ and assume that there is one type of organic bug that exists with probability $\mu=\frac{1}{2}$, has complexity $q=\frac{1}{2}$, and that finding the bug derives a utility of $w=2$ for the designer. For simplicity, we assume $n=2$ agents. Then, the optimal threshold given by the fixed point of \eqref{eq:foc} is $\tilde{c}=\frac{2}{9}$. The critical value depending on the budget constraint $\overline{v}$ is $c_0(\overline{v})=\left(\frac{4}{\overline{v}}+\frac{1}{4} \right)^{-1}$. Thus, the condition $\tilde{c}>c_0(\overline{v})$ from \Cref{thm:artificialbug} holds if and only if $\overline{v}<\frac{16}{17}$, i.e. inserting an artificial bug is beneficial if and only if the budget is constrained by at least $\frac{16}{17}$.

For example, if $\overline{v}=\frac{1}{2}$ and if we do not allow the insertion of any artificial bug, the maximum is attained at $c_0(\overline{v})=\frac{4}{33}$, leading to a utility of $\frac{32}{363} \approx 0.088$ for the designer. With an artificial bug, we can attain the optimal threshold of $\tilde{c}=\frac{2}{9}$ by choosing, for example, $v=\frac{4}{14}$, $v_a=\frac{3}{14}$ and $q_a=1$, so that we obtain a utility of $\frac{1}{9} \approx 0.111 $.

Next, we illustrate graphically how inserting an artificial bug lowers the designer's financial commitment, in other words allows spending less on the prizes $v$ and $v_a$. Consider one organic and one artificial bug. In \Cref{fig:q_a}, we fix $\mu$, $q$, and plot the optimal prize sets $\left \{(v, v_a): \tilde{c}=v \mu \Phi(\tilde{c}; q) + v_a \Phi(\tilde{c};q_a) \right \}$
for different $q_a$'s on which the maximal utility is achieved, i.e., at the optimal equilibrium threshold, $\tilde{c}$. The black line labeled $v + v_a = \overline{v}$ represents the budget constraint for $\overline{v}=\frac{2}{3}$. Therefore, the optimal equilibrium threshold is achievable if the prizes $(v, v_a)$ lie within the budget set. Observe that having a less complex artificial bug, i.e., higher $q_a$, leads to a larger portion of the optimal prize set being within the budget set. In contrast, if the artificial bug is too complex to find, for example, if $q_a=\frac{1}{5}$, we cannot achieve the optimal threshold.

\begin{figure}[ht]
    \centering
    \begin{minipage}{.5\textwidth}
        \centering
        \includegraphics[width=1\linewidth]{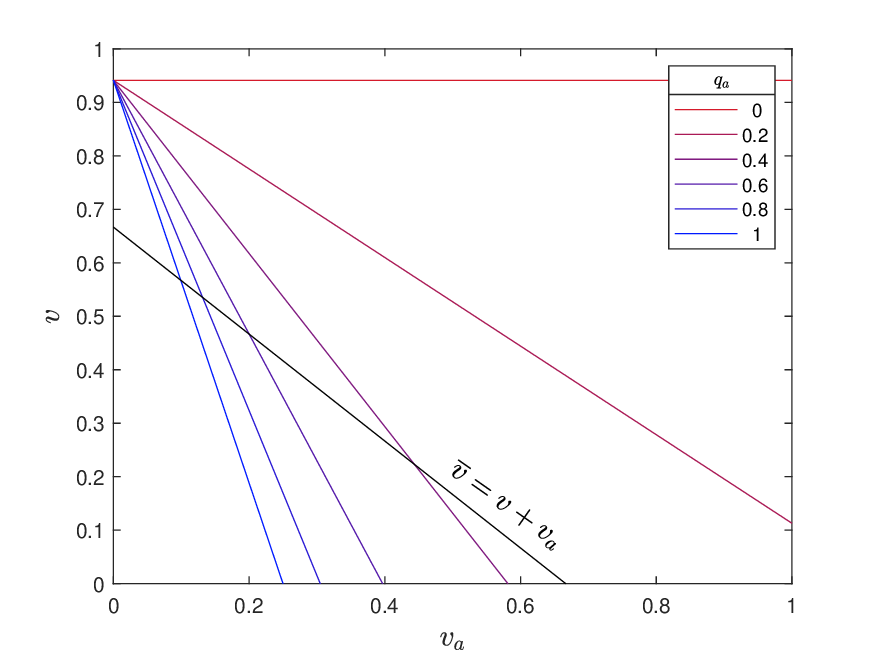}
        \captionof{figure}{Optimal $(v, v_a)$ given $q_a$}
        \label{fig:q_a}
    \end{minipage}%
    \begin{minipage}{.5\textwidth}
        \centering
        \includegraphics[width=1\linewidth]{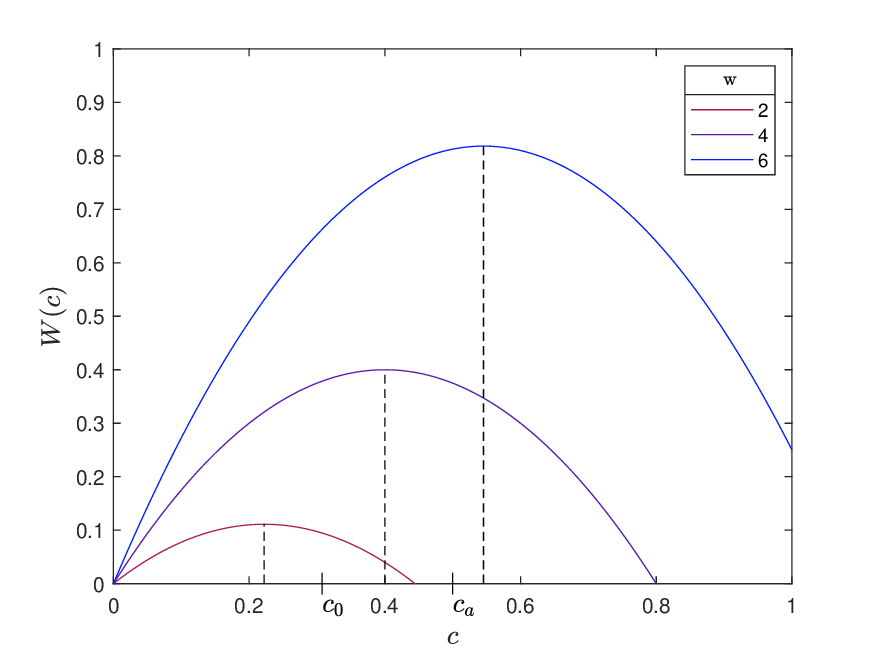}
        \captionof{figure}{Dependence on $w$}
        \label{fig:w}
    \end{minipage}
\end{figure}

In \Cref{fig:w}, we plot the utility $W(\hat{c})$ of the designer against the thresholds $\hat{c}$ for different $w$'s. Observe that if we do not allow any artificial bugs, we can achieve only thresholds in $[0, c_0(\overline{v})] \approx [0, 0.308]$, whereas an artificial bug helps to achieve thresholds up to $c_a(\overline{v})=\frac{1}{2}$. For low $w$ (take, for example, $w=2$), inserting an artificial bug does not help to achieve the maximum. If we increase $w$, an artificial bug is needed to achieve the maximal utility (say $w=4$). Increasing $w$ further ($w=6$) can lead to the case that even when inserting an artificial bug, we are unable to achieve the maximal utility. Nevertheless, inserting an artificial bug leads to a greater utility than without one. Note that we obtain the same observations if we vary the probability $\mu$ that a bug exists instead of $w$.

\subsection{Public Program}

To illustrate results for the public program, consider the uniform distribution on $[1,2]$. As before, we assume that there is one type of bug that exists with probability $\mu=\frac{1}{2}$, has complexity $q=\frac{1}{2}$, and that finding it derives a utility of $w=10$ for the designer. The budget constraint is set to $\overline{v}=5$. \Cref{fig:utility_convergence} is a visualization of the convergence of the designer's objective function $W_n(\hat{c})$ as $n \rightarrow \infty$ as seen in \Cref{cor:public_utility}. Observe that the larger the number of agents, the smaller the optimal threshold $c_n$ as $c_n \rightarrow \underline{c} = 1$ from \Cref{lemma:limits}. \Cref{fig:utility_convergence_scaled} depicts the same plot, but with the $x$-axis scaled by $nF(\hat{c})$, to visualize the convergence of $nF(c_n) \rightarrow \kappa^*$. Note that we can also see the convergence of the designer's objective function as $n \rightarrow \infty.$

\begin{figure}[ht]
    \centering
    \begin{minipage}{0.5\textwidth}
        \centering
        \includegraphics[width=1\linewidth]{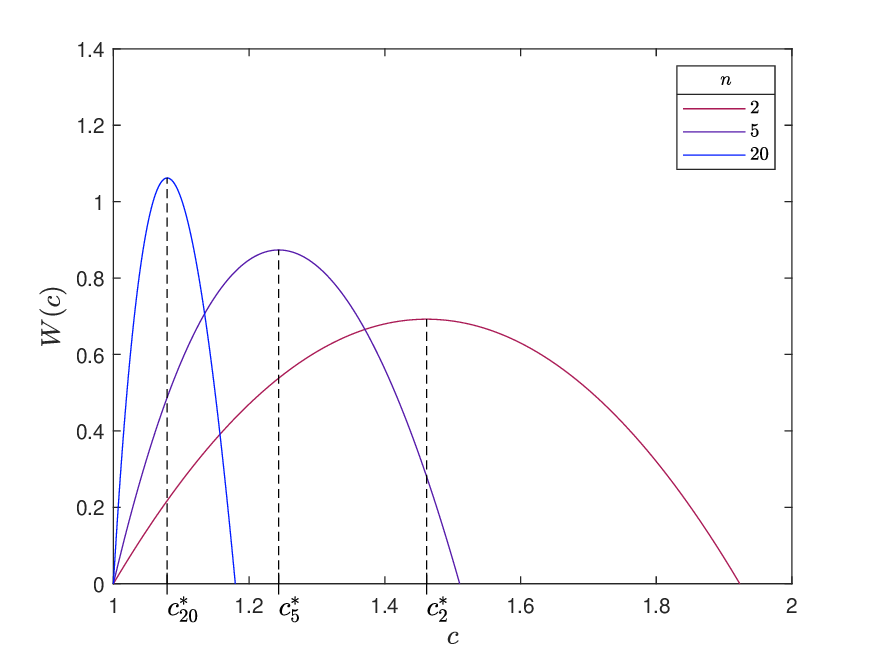}
        \captionof{figure}{Convergence of $W_n(\hat{c})$ as $n$ grows}
        \label{fig:utility_convergence}
    \end{minipage}%
    \begin{minipage}{0.5\textwidth}
        \centering
        \includegraphics[width=1\linewidth]{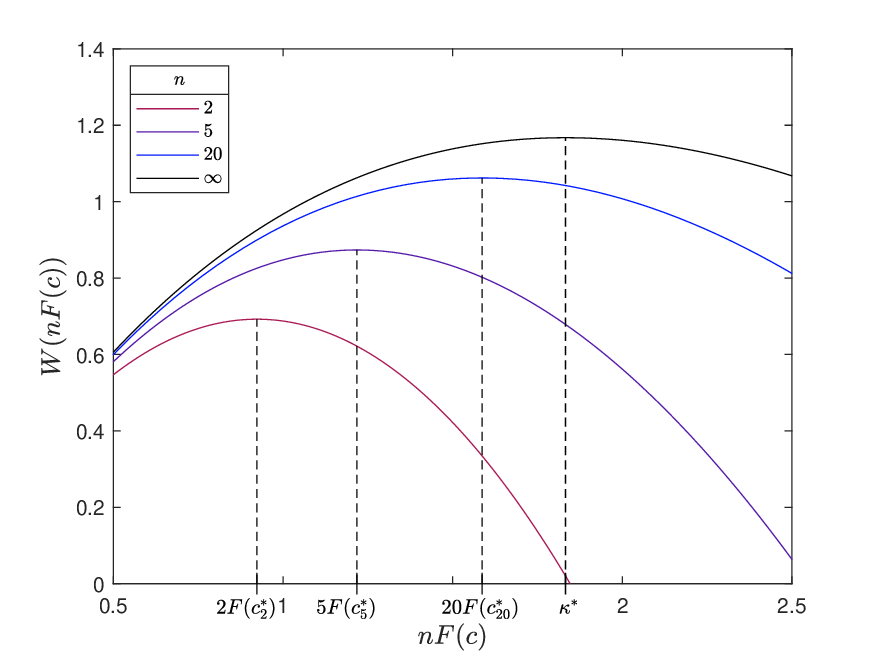}
        \captionof{figure}{Scaled version}
        \label{fig:utility_convergence_scaled}
    \end{minipage}%
\end{figure}

The effect of the budget constraint $\overline{v}$ and the artificial bug $q_a$ in the public program is analogous to those in the private program. We therefore omit its visualization and refer to \Cref{fig:q_a} and \Cref{fig:w} in the previous section. 

Instead, we illustrate the convergence of the set of solutions for the private program with $n$ agents to the solutions of the public program from \Cref{thm:limit_set}. In \Cref{fig:set_convergence}, we show projections of the sets $M_n$ and $M_\infty$ for a fixed $q_a \in \left\{ \frac{1}{3}, \frac{1}{2}, 1 \right\}$. We observe that the Hausdorff distance decreases as $n$ increases. Another interesting observation is that with artificial bugs that are less complex to find, the designer gains more flexibility in setting the prizes in the sense that the set of optimal prizes enlarges and allows less spending.

\begin{figure}[h]
    \makebox[\textwidth][c]{\includegraphics[width=1.2\textwidth]{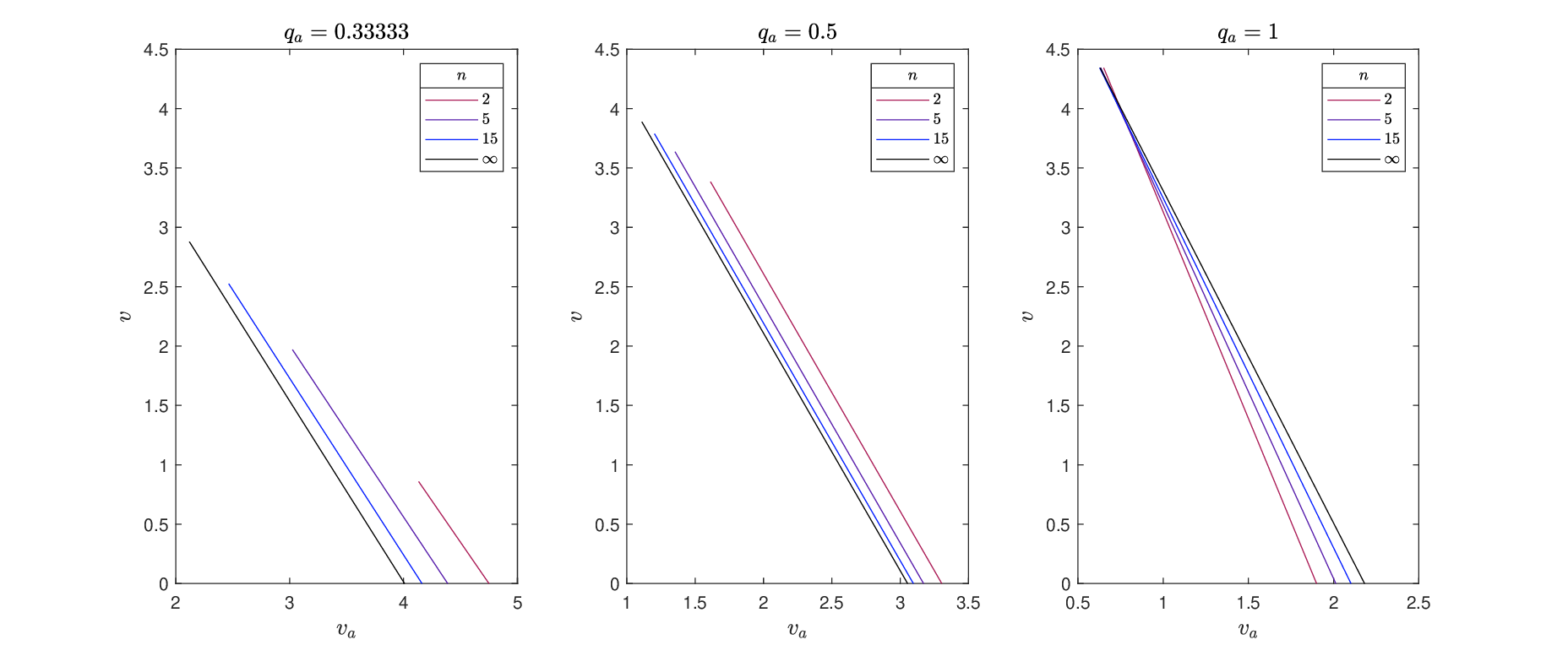}}%
    \captionof{figure}{Convergence of the sets $M_n$ to $M_\infty$}
    \label{fig:set_convergence}
\end{figure}

%%%%%%%%%%%%%%%%%%%%%%%%%%%%%%%%%%%%%%%%%%%%%%%%%%%%%%%%%%%%%%%%%%%%%%
%%%%%%%%%%%%%%%%%%%%%%%%%%%%%%%%%%%%%%%%%%%%%%%%%%%%%%%%%%%%%%%%%%%%%%

\section{Implementation} \label{sec:implement}

The results in the previous sections show that artificial bugs are a welcome instrument for the designer of bug bounty programs, especially when finding organic bugs is important for the designer. Using artificial bugs in the design of bug bounty programs does not pose any particular complexity from a technical perspective, as one can easily introduce simpler or more sophisticated bugs in software packages. Yet, as prizes paid for finding artificial and organic bugs may optimally differ, the designer may want to prove to the finders of the artificial bug, or even to all participants, that an artificial bug found was indeed inserted on purpose and was artificially designed by the designer at the start of the bug bounty program. Even more importantly, if the artificial bug is not found during the crowdsearch, it is important that the designer can prove that an artificial bug has been inserted before the crowdsearch started. This would ensure, or reaffirm, the credibility of the bug bounty program with artificial bugs. We outline three approaches that could be used to achieve this objective: encryption, commitment scheme, and zero-knowledge (ZK) proofs.\footnote{Of course, the credibility of inserting artificial bugs could also be ensured by traditional means. For instance, the designer could invite a notary, who confirms in a written statement that an artificial bug has been inserted.}

Let us start with asymmetric encryption (see, e.g., \citet{stallings2020}). The original block of the code and the modified code block with the artificial bug could be encrypted before the crowdsearch starts, and participants can decrypt it once the crowdsearch has ended. More specifically, the designer and participants are connected to a trusted entity that generates a public key and a private key for each participant. The public key information is stored in a directory. Before the crowdsearch starts, the designer encrypts both the original and the modified block of code with the artificial bug using the participant's public key and sends the encrypted message to all participants. Once the crowdsearch competition has ended, the participants receive their private keys to decrypt the message and verify the existence of the artificial bug.

This approach also allows the artificial bug to only exist with a particular probability, which is common knowledge to the designer and all participants before the crowdsearch starts. The trusted entity flips a biased coin and, based on the outcome of the coin, sends a message to the designer about whether to insert a bug or not. In case no bug should be introduced, the private keys will be empty and not be usable to decrypt any message. This shows to the participants that no artificial bug has been inserted. We note that such an approach is incentive-compatible for the bug bounty designer. If no artificial bug should be inserted, there is also no benefit in inserting such a bug secretly, as such a bug would be treated as an organic bug. If detected, the designer has to pay, and thus, inserting a bug would only be worse off, as the incentive to search for bugs is unaffected.

While the bug bounty platform is a natural candidate for the trusted entity, an alternative approach that does not rely on having a third party is for the designer to use a commitment scheme \citep{goldreich2001}. Specifically, before the crowdsearch starts, the designer computes a commitment over the modified block of code with the artificial bug. The designer then publishes the commitment without revealing any further information about the artificial bug (\textit{hiding} property). Once the competition has ended, the designer then opens the commitment and proves to the participants that the designer knows that the modified code block contains an artificial bug (\textit{binding} property).

Inserting an artificial bug with a certain probability can also be implemented with the commitment scheme but needs a public source of randomness and thus ultimately also a trusted third party (see \citet{bonneau2022} for an overview). To implement this, the designer commits to a private key and flips a coin based on it and a public source of randomness. At the end of the crowdsearch, the designer opens the commitment and allows participants to verify the result of the coin flips and that the protocol was followed.

A third approach to reveal an artificial bug convincingly to the participants\textemdash irrespective of whether it has been found or not\textemdash is to use zero-knowledge (ZK) proofs. The idea is that before the start of a crowdsearch, the designer demonstrates the existence of a artificial bug in the software, without revealing the location and underlying technique of the bug. Hence, the designer can convince the participants of the artificial bug's existence without giving any hint on how to find it.

While this concept is attractive in theory, proving bugs in ZK in practical software requires solving a variety of problems and engineering challenges. As pointed out by \citet{cuellar2023}, ZK frameworks must be able to compile proofs of bugs that require many steps of execution. Moreover, it is necessary to efficiently create understandable statements. Yet, recent advances in the development of proof-statement compilers show promising progress as to how the existence of a bug can be proved without revealing the details or the inputs that could be used to exploit the bug.\footnote{See \citet{cuellar2023,green2023} and the survey on the applicability of ZKs in various cases, \citet{ernstberger2024}.} For our problem, a ZK proof is applicable, as the bug can be chosen to be particularly suitable to the applications of the proof-statement compiler and supporting tools described in \citet{cuellar2023} and \citet{green2023}. Hence, the inserted bug can be chosen such that a ZK proof is available. Despite its requirements for sophisticated compilers, a key benefit of ZK proofs over the previous two approaches is that participants can verify \textit{at the start} of the crowdsearch that an artificial bug exists rather than \textit{at the end}.

While the three presented approaches can provide an implementation of artificial bugs as envisioned, which approach is most useful will depend on the specific applications, as well as on the reputation and engineering capabilities of the bug bounty designer.

%%%%%%%%%%%%%%%%%%%%%%%%%%%%%%%%%%%%%%%%%%%%%%%%%%%%%%%%%%%%%%%%%%%%%%
%%%%%%%%%%%%%%%%%%%%%%%%%%%%%%%%%%%%%%%%%%%%%%%%%%%%%%%%%%%%%%%%%%%%%%

\section{Other Benefits of Artificial Bugs} \label{sec:discuss}

Our paper suggests that inserting artificial bugs can be useful to the designer of bug bounty programs. We have examined one such benefit, showing that artificial bugs allow the designer to incentivize participants at a lower cost. We conclude the paper by discussing other potential benefits of artificial bugs.

First, artificial bugs can help screen invalid submissions, which saves the organization costly resources on triage and verification. Bug bounty programs typically attract large numbers of invalid submissions as security researchers incur relatively small marginal costs to submit a report. \citet{zhao2017} report that across major bug bounty platforms, less than 25\% of submissions are valid, creating tremendous inefficiency in the verification efforts. Organizations can significantly reduce this inefficiency by prioritizing submissions from researchers who also reported artificial bugs. In other words, artificial bugs can serve as a badge for quality submissions.

Second, artificial bugs can be used to gauge participation in bug bounty programs. Organizations receive submissions of vulnerabilities but often cannot determine the number of researchers currently engaging in their bug bounty program. To remedy this uncertainty, artificial bugs can be used as entry checks or milestones in determining the number of active researchers, as well as the extent to which they have probed the system.

Third, artificial bugs can be used to renew interest in the bug bounty program. An empirical analysis of the HackerOne platform shows that bug bounty programs receive less engagement\textemdash as proxied by the number of submissions\textemdash over time, since researchers are attracted to newly published programs \citep{sridhar2021}.  \citet{maillart2017} show that this front-loading effect results in the probability of finding bugs decaying at a rate $1/t^{0.4}$, where $t$ is the time since the program was launched. Inserting artificial bugs may help enhance interest in an otherwise unattractive program.

\section{Conclusion} \label{sec:conclusion}

We have developed a simple device to improve the effectiveness of bug bounty systems.  The additional benefits outlined in the last section strengthen the proposal and establish promising directions for future research into modeling, utilizing, and implementing the insertion of artificial bugs into a bug bounty system.

%%%%%%%%%%%%%%%%%%%%%%%%%%%%%%%%%%%%%%%%%%%%%%%%%%%%%%%%%%%%%%%%%%%%%%
%%%%%%%%%%%%%%%%%%%%%%%%%%%%%%%%%%%%%%%%%%%%%%%%%%%%%%%%%%%%%%%%%%%%%%
%%%%%%%%%%%%%%%%%%%%%%%%%%%%%%%%%%%%%%%%%%%%%%%%%%%%%%%%%%%%%%%%%%%%%%
%%%%%%%%%%%%%%%%%%%%%%%%%%%%%%%%%%%%%%%%%%%%%%%%%%%%%%%%%%%%%%%%%%%%%%

\clearpage

\singlespacing

% % That's the necessary style for economics publications
\bibliographystyle{apalike}
% % You can add several .bib-file with comma as separator
\bibliography{bug}

%\clearpage

%\linenumbers
%\singlespacing
%\onehalfspacing
\doublespacing

\appendix

%%%%%%%%%%%%%%%%%%%%%%%%%%%%%%%%%%%%%%%%%%%%%%%%%%%%%%%%%%%%%%%%%%%%%%
%%%%%%%%%%%%%%%%%%%%%%%%%%%%%%%%%%%%%%%%%%%%%%%%%%%%%%%%%%%%%%%%%%%%%%

\section{Proofs} \label{sec:proofs}

\begin{proof}[Proof of \Cref{lemma:equilibrium}]
First, we show that every equilibrium strategy is a threshold strategy. Let $\sigma_i: [\underline{c},\overline{c}] \rightarrow \{0,1\}$ denote agent $i$'s strategy and let $\boldsigma_{-i}$ denote the strategy profile of agents other than $i$. Given a strategy profile, the payoff of agent $i$ is
\begin{equation*}
u_i(\sigma_i,\boldsigma_{-i},c_i) = \sigma_i\left( \sum_{l} v^l\mu^l p(\boldsigma_{-i};q^l) + \sum_{k} v_a^k p(\boldsigma_{-i};q_a^k) - c_i\right),\end{equation*}
where $p(\boldsigma_{-i};q)$ is the probability that agent $i$ wins the prize corresponding to the bug with complexity $q$, given the strategy $\boldsigma_{-i}$ of others.

Let $\sigma_i^*$ be an equilibrium strategy for agent $i$. We prove by contradiction that $\sigma_i^*$ is non-increasing. Suppose that there exists a pair of costs $c < c'$ such that $\sigma_i^*(c)=0$ and $\sigma_i^*(c')=1$. Given the private cost vector $\boldc =(c_1, \dots, c_n)$,
\begin{equation*}
\begin{aligned}
    0 & = \mathbb{E}[u_i(0, \sigma^*_{-i}(c_{-i}), c_i)|c_i=c] \geq \mathbb{E}[u_i(1, \sigma^*_{-i}(c_{-i}, c_i)|c_i=c] \\
    & >  \mathbb{E}[u_i(1, \sigma^*_{-i}(c_{-i}), c_i)|c_i=c'] \geq \mathbb{E}[u_i(0, \sigma^*_{-i}(c_{-i}, c_i)|c_i=c'] = 0
\end{aligned}
\end{equation*}
where the first and the last inequalities follow from the definition of the equilibrium. The strict inequality, which follows from the fact that $u_i$ is strictly decreasing in $c$ leads to a contradiction. Thus, any equilibrium strategy is a threshold strategy.

Furthermore, since we consider symmetric equilibrium, all agents use the same threshold strategy denoted by $c^*$. In equilibrium, the agent with search cost at the threshold must be indifferent between searching and not searching. Therefore, if the equilibrium threshold lies in the interior of $[\underline{c}, \overline{c}]$, it must hold that
\begin{equation*}
    c^* = \sum_l v^l q^l \Phi(c^*; q^l) + \sum_k v_a^k \Phi(c^*; q_a^k),
\end{equation*}
where $\Phi(\hat{c};q)$ is the expectation of $p(\boldsigma_{-i}(\boldc_{-i});q)$, given the cost distribution of the other agents when they all use the threshold.

To conclude, note that $\Phi(\hat{c};q)$ is decreasing in  $\hat{c}$ for every $q \in (0, 1]$ because a higher threshold adopted by the other agents increases their probability to search and find the bug, which, in turn, must lower agent $i$'s probability of being rewarded for that particular bug. It follows that $\Psi$ is also decreasing in $\hat{c}$. Thus, if $\Psi(\underline{c})\leq \underline{c}$, every agent can only expect to incur a loss by participating in the search regardless of his/her skills, and hence $c^*=\underline{c}$. If $\Psi(\overline{c})\geq \overline{c}$, every agent can only profit by participating, and therefore $c^*=\overline{c}$. Otherwise, we have that $\Psi(\underline{c}) < \underline{c}$ and $\overline{c}<\Psi(\overline{c})$. By continuity and monotonicity of $\Psi$, there exists a unique fixed point. This completes the proof.

\end{proof}

%%%%%%%%%%%%%%%%%%%%%%%%%%%%%%%%%%%%%%%%%%%%%%%%%%%%%%%%%%%%%%%%%%%%%%
%%%%%%%%%%%%%%%%%%%%%%%%%%%%%%%%%%%%%%%%%%%%%%%%%%%%%%%%%%%%%%%%%%%%%%

\begin{proof}[Proof of \Cref{lemma:onebug}]
    Suppose that $(\boldv^*, v_a^*, q_a^*)$ solves the designer's problem \eqref{eq:designer_problem} and assume w.l.o.g. that $q_a^* \neq 0$ (else there is nothing to prove). Recall that the equilibrium threshold is given by
    \begin{equation}\label{eq:threshold_proof_lem2}
        c^* = \sum_{l} v^{l*}\mu^l \Phi(c^*;q^l) + \sum_{k} v_a^{k*} \Phi(c^*;q_a^{k*}).
    \end{equation}
    Introduce $\tilde{k} \equiv \argmax_{k} \{q_a^{k*}\}$ and choose the new complexities $\boldq_a^{**}=(q_a^{1**}, 0, \dots, 0)$ with $q_a^{1**}=q_a^{\tilde{k}*}$. Moreover let $\boldv_a^{**}=(v_a^{1**}, 0, \dots, 0)$ where the first prize is given by
    \begin{equation*}
        v_a^{1**}=v_a^{\tilde{k}*}+\frac{\sum_{k\neq \tilde{k}} v_a^{k*} P(c^*; q_a^{k*})}{P(c^*; q_a^{1**})}.
    \end{equation*}
    Observe that with this choice, the second sum in equation \eqref{eq:threshold_proof_lem2} does not change, i.e.
    \begin{equation}\label{eq:stable_sum_artbug}
         \sum_{k} v_a^{*k} \Phi(c^*;q_a^{*k}) =  \sum_{k} v_a^{**k} \Phi(c^*;q_a^{**k}) = v_a^{1**} \Phi(c^*;q_a^{1**}).
    \end{equation}
    Thus if we take $\boldv^{**}=\boldv^{*}$ we obtain the same equilibrium threshold and therefore the designer receives the same utility because of equation \eqref{eq:stable_sum_artbug}. Note that the boundary conditions of \eqref{eq:designer_problem} are satisfied since $P(c^*; q)$ is increasing in $q$. Hence, it is sufficient to introduce only one artificial bug.
\end{proof}

%%%%%%%%%%%%%%%%%%%%%%%%%%%%%%%%%%%%%%%%%%%%%%%%%%%%%%%%%%%%%%%%%%%%%%
%%%%%%%%%%%%%%%%%%%%%%%%%%%%%%%%%%%%%%%%%%%%%%%%%%%%%%%%%%%%%%%%%%%%%%

\begin{proof}[Proof of \Cref{lemma:keyrelation}]
Recall that $\Phi(\hat{c}; q)$ is the probability that an agent wins the prize for the bug with complexity $q$ when participating. Since agents find a given bug independently of other agents, we can write
\begin{align*}
    \Phi(\hat{c};q) & = q \sum_{k=0}^{n-1} \left\{ {n \choose k} F(\hat{c})^k(1-F(\hat{c}))^{n-1-k} \left [ \sum_{t=0}^k {k \choose t} \frac{q^t(1-q)^{k-t}}{t+1} \right] \right\} \\
    & = \frac{1-(1-q F(\hat{c}))^n}{n F(c)} = \frac{P(\hat{c};q)}{n F(\hat{c})},
\end{align*}
where we used a generalization of the binomial theorem twice. For details, see the proof of Proposition 3 and Fact 1 in \citet{gersbach2023}.
\end{proof}

%%%%%%%%%%%%%%%%%%%%%%%%%%%%%%%%%%%%%%%%%%%%%%%%%%%%%%%%%%%%%%%%%%%%%%
%%%%%%%%%%%%%%%%%%%%%%%%%%%%%%%%%%%%%%%%%%%%%%%%%%%%%%%%%%%%%%%%%%%%%%

\begin{proof}[Proof of \Cref{lemma:achievablec}]
    Observe that the set of achievable equilibrium thresholds $C(\overline{v})$
    has to be an interval since $\sum_{l} v^l\mu^l \Phi(\hat{c};q^l) + v_a \Phi(\hat{c};q_a)$ is continuous. The smallest possible value for $\hat{c}$ is $0$, which can be achieved by choosing $v_a=0$, $q_a=0$ and $v^l=0 \quad \forall l$. Moreover for any $l$ we know from \Cref{lemma:keyrelation} that
    \begin{equation*}
        \Phi(\hat{c}; q^l) = \frac{P(\hat{c}; q^l)}{n F(\hat{c})} = \frac{1-(1-q^lF(\hat{c}))^n}{n F(\hat{c})}\leq \frac{1-(1-F(\hat{c}))^n}{n F(\hat{c})} = \Phi(\hat{c}; 1).
    \end{equation*}
    This, together with the fact that $\Phi(\hat{c}, 1)$ is decreasing in $\hat{c}$ (see Proof of \Cref{lemma:keyrelation}), tells us that the highest $\hat{c} \in C(\overline{v})$ is achieved by allocating all of the reward $\overline{v}$ to the finding of an artificial bug with complexity $q_a=1$, i.e. $\max C(\overline{v})=c_a(\overline{v})$ with $c_a(\overline{v})$ the fixed point of $\overline{v} \Phi(\hat{c}, 1)$. Thus $C(\overline{v})=[0, c_a(\overline{v})]$.
\end{proof}

%%%%%%%%%%%%%%%%%%%%%%%%%%%%%%%%%%%%%%%%%%%%%%%%%%%%%%%%%%%%%%%%%%%%%%
%%%%%%%%%%%%%%%%%%%%%%%%%%%%%%%%%%%%%%%%%%%%%%%%%%%%%%%%%%%%%%%%%%%%%%

\begin{proof}[Proof of \Cref{thm:opt}]
    Recall that we can write the designer's problem as
    \begin{equation*}
        \max_{\hat{c} \in [0,c_a(\overline{v})]} \; \sum_{l} w^l\mu^l P(\hat{c};q^l) - n F(\hat{c})\hat{c}.
    \end{equation*}
    Taking the derivative with respect to $\hat{c}$ leads to the first-order condition
    \begin{equation*}
    \begin{aligned}
        \frac{d}{d\hat{c}} \left(\sum_l w^l\mu^l P(\hat{c}; q^l) - n F(\hat{c}) \hat{c} \right) &= \sum_l  \frac{d}{d\hat{c}} \left(w^l \mu^l (1-(1-q^l F(\hat{c}))^n) - n F(\hat{c}) \hat{c} \right) \\
        &= \sum_l w^l \mu^l n(1-q^l F(\hat{c}))^{n-1})q^l f(\hat{c}) - n f(\hat{c}) \hat{c} - n F(\hat{c}) \\ &= 0,
    \end{aligned}
    \end{equation*}
    where we used that $P(\hat{c}; q^l)=1-(1-q^lF(\hat{c}))^n$. By rearranging, we see that
    \begin{equation*}
         \hat{c}= \sum_l w^l \mu^l q^l (1-q^lF(\hat{c}))^{n-1} - \frac{F(\hat{c})}{f(\hat{c})} = \Omega(\hat{c}),
    \end{equation*}
    i.e. $\hat{c}$ is a fixed point of $\Omega$. Note that $(1-q^l F(\hat{c}))^{n-1}$ is decreasing in $\hat{c}$ and recall that we assumed $F/f$ is non-decreasing, and therefore $\Omega({\hat{c}})$ is decreasing in $\hat{c}$. Hence, $\Omega(\hat{c})$ has a unique fixed point $\tilde{c}$. 
    
    This fixed point has to be a maximum since the derivative is $\frac{d W(\hat{c})}{d \hat{c}}= n f(\hat{c}) ( \Omega(\hat{c}) - \hat{c})$, and therefore $\frac{d W}{d \hat{c}}(\hat{c})>0$ for $\hat{c}<\tilde{c}$ and $\frac{d W}{d\hat{c}} (\hat{c}) <0$ for $\tilde{c}<\hat{c}$. By \Cref{lemma:achievablec}, we know that the set of achievable thresholds is $C(\overline{v}) = [0,c_a(\overline{v})]$ and since the utility is increasing on $[0, \tilde{c}]$, the optimal achievable equilibrium threshold is $\hat{c}^*=\min\{\tilde{c}, c_a(\overline{v})\}$. Using \Cref{lemma:equilibrium}, we can conclude that any prizes $v^*, v_a^*$ and complexity $q_a^*$ that solve $\hat{c}^* = \sum_{l}v^l \mu^l \Phi(\hat{c}^*;q^l) + v_a \Phi(\hat{c}^*;q_a)$ are optimal.
\end{proof}

%%%%%%%%%%%%%%%%%%%%%%%%%%%%%%%%%%%%%%%%%%%%%%%%%%%%%%%%%%%%%%%%%%%%%%
%%%%%%%%%%%%%%%%%%%%%%%%%%%%%%%%%%%%%%%%%%%%%%%%%%%%%%%%%%%%%%%%%%%%%%

\begin{proof}[Proof of \Cref{thm:artificialbug}]
    This theorem follows from the observation that not inserting an artificial bug reduces the set of achievable equilibrium thresholds to $C_0(\overline{v})=[0, c_0(\overline{v})] \subseteq C(\overline{v})$. Since by \Cref{thm:opt}, the optimal equilibrium threshold is given by $\hat{c}^*=\min\{\tilde{c}, c_a(\overline{v})\}$ and $c_0(\overline{v}) \leq c_a(\overline{v})$, this affects the optimal solution if and only if $\tilde{c} > c_0(\overline{v})$.
\end{proof}

%%%%%%%%%%%%%%%%%%%%%%%%%%%%%%%%%%%%%%%%%%%%%%%%%%%%%%%%%%%%%%%%%%%%%%
%%%%%%%%%%%%%%%%%%%%%%%%%%%%%%%%%%%%%%%%%%%%%%%%%%%%%%%%%%%%%%%%%%%%%%

\begin{proof}[Proof of \Cref{lemma:limits}]
    \begin{enumerate}
        \item[(i)] 
            Recall from \Cref{lemma:equilibrium} that for any bug with complexity $q$ we have that $P_n(\hat{c}; q)=nF(\hat{c}) \Phi_n(\hat{c}; q)$. Thus the equilibrium condition can be written as
            \begin{equation}\label{eq:equilibrium_rewritten}
                c_n n F(c_n)= \sum_l v^l \mu^l P_n(c_n; q^l) + v_a P_n(c_n; q_a).
            \end{equation}
            Assume by contradiction that $\lim_{n \rightarrow \infty} c_n $ is not equal to $\underline{c}$. By using a generalized binomial theorem, one can prove that $\Psi_n(c)$ is decreasing in $n$. Therefore, $c_n$ has to be decreasing in $n$. Hence, there has to exist some value $c' > \underline{c}$ such that $c_n > c'$ for any $n \in \mathbb{N}$. In this case, $c_n n F(c_n) > c' n F(c')$. This holds because $F$ is strictly increasing. Note that $F(c')>0$, since $c'>\underline{c}$. Therefore, $\lim_{n\rightarrow \infty}{c_n n F(c_n)}\geq \lim_{n\rightarrow \infty}c'nF(c')=\infty$, which cannot be equal to \[\lim_{n \rightarrow \infty}\sum_l v^l \mu^l P_n(c_n; q^l) + v_a P_n(c_n; q^a) \leq \sum_l v^l + v_a,\] i.e., we obtain a contradiction.
        \item[(ii)]
            For the second part, we assume that $n F(c_n)$ converges as $n \rightarrow \infty$ and denote the limit by $\tilde{\kappa}$.\footnote{It can be proven that $n F(c_n)$ indeed converges by using Taylor series expansions (see Proposition 9 in \cite{gersbach2023}).} Next define $\kappa^* \in (0, \infty)$ as the unique fixed point of $\Psi_{\infty}(\hat{\kappa})=\sum_l v^l \mu^l \frac{1-e^{-q^l \hat{\kappa}}}{\underline{c}} + v_a \frac{1 - e^{-q_a \hat{\kappa}}}{\underline{c}}$. The fixed point exists since $\Psi_{\infty}(0)=0$ and $\Psi_{\infty}(\hat{\kappa})$ is bounded by $\sum_l v^l + v_a$. Uniqueness follows from the fact that $\Psi_{\infty}(\hat{\kappa})$ is strictly increasing in $\hat{\kappa}$. Thus, we only need to prove that $\kappa^*=\tilde{\kappa}$. Note that for any complexity $q \in [0, 1]$ it holds that
            \begin{equation}\label{eq:limit_p}
                \lim_{n \rightarrow \infty} (1-qF(c_n))^n = \lim_{n \rightarrow \infty} \left[ \left ( 1 - \frac{q}{1/F(c_n)} \right )^{\frac{1}{F(c_n)}} \right ]^{nF(c_n)} = e^{-q \tilde{\kappa}},
            \end{equation}
            by the fact that $F(c_n) \rightarrow F(\underline{c})=0$, the definition of $e^x$, and the continuity of the function. Using this observation together with equation (\ref{eq:equilibrium_rewritten}) from part (i) leads to
            \begin{equation*}
                \begin{aligned}
                    \underline{c} \tilde{\kappa} =& \lim_{n \rightarrow \infty } c_n n F(c_n) = \lim_{n \rightarrow \infty} \sum_l v^l \mu^l  P_n(c_n; q^l) + v_a  P_n(c_n; q_a) \\
                    =& \sum_l v^l \mu^l \lim_{n \rightarrow \infty} \left(1 - (1 - q^l F(c_n))^n \right) + v_a \lim_{n \rightarrow \infty} \left( 1 - (1 - q_a F(c_n))^n \right) \\
                    =& \sum_l v^l \mu^l (1- e^{-q^l \tilde{\kappa}}) + v_a (1- e^{-q_a \tilde{\kappa}}).
                \end{aligned}
            \end{equation*}
            In particular, $\tilde{\kappa}$ is a fixed point of $\Psi_{\infty}(\hat{\kappa})$ and by uniqueness we can conclude that $\tilde{\kappa}=\kappa^*$.
        \item[(iii)]
            The proof of $\lim_{n \rightarrow \infty} P_n(q) = 1 - e^{-q \kappa^*}$ follows directly from equation (\ref{eq:limit_p}).
    \end{enumerate}
\end{proof}

%%%%%%%%%%%%%%%%%%%%%%%%%%%%%%%%%%%%%%%%%%%%%%%%%%%%%%%%%%%%%%%%%%%%%%
%%%%%%%%%%%%%%%%%%%%%%%%%%%%%%%%%%%%%%%%%%%%%%%%%%%%%%%%%%%%%%%%%%%%%%

\begin{proof}[Proof of \Cref{cor:public_utility}]
    The corollary follows directly from \Cref{lemma:limits}.
\end{proof}

%%%%%%%%%%%%%%%%%%%%%%%%%%%%%%%%%%%%%%%%%%%%%%%%%%%%%%%%%%%%%%%%%%%%%%
%%%%%%%%%%%%%%%%%%%%%%%%%%%%%%%%%%%%%%%%%%%%%%%%%%%%%%%%%%%%%%%%%%%%%%

\begin{proof}[Proof of \Cref{lemma:achievable}]
    Observe that 
    \begin{equation*}
        \Psi_\infty(\hat{\kappa}; \boldv, v_a, q_a)=\sum_l v^l \mu^l \frac{1-e^{-q^l \hat{\kappa}}}{\underline{c}} + v_a \frac{1-e^{-q_a \hat{\kappa}}}{\underline{c}}
    \end{equation*}
    is increasing in $\hat{\kappa}$ and $q_a$. Thus, the lemma can be proven analogously to \Cref{lemma:achievablec}.
\end{proof}

%%%%%%%%%%%%%%%%%%%%%%%%%%%%%%%%%%%%%%%%%%%%%%%%%%%%%%%%%%%%%%%%%%%%%%
%%%%%%%%%%%%%%%%%%%%%%%%%%%%%%%%%%%%%%%%%%%%%%%%%%%%%%%%%%%%%%%%%%%%%%

\begin{proof}[Proof of \Cref{thm:public_opt}]
    As in the proof of \Cref{thm:opt}, we use the first-order condition
    \begin{equation*}
        \frac{d}{d\hat{\kappa}} \left( \sum_l w^l \mu^l  \left(1- e^{-q^l \hat{\kappa}} \right)- \hat{\kappa} \underline{c} \right) = \sum_l w^l \mu^l q^l e^{-q^l \hat{\kappa}} - \underline{c} = \Omega_\infty(\hat{\kappa}).
    \end{equation*}
    By assumption, we have $\Omega_\infty(0)\geq 0$ and $\Omega_\infty(\hat{\kappa}) \rightarrow - \underline{c}$ as $\hat{\kappa} \rightarrow \infty$, we obtain an unique $\tilde{\kappa}$. The rest of the proof uses the same arguments as the proof of \Cref{thm:opt}.
\end{proof}

%%%%%%%%%%%%%%%%%%%%%%%%%%%%%%%%%%%%%%%%%%%%%%%%%%%%%%%%%%%%%%%%%%%%%%
%%%%%%%%%%%%%%%%%%%%%%%%%%%%%%%%%%%%%%%%%%%%%%%%%%%%%%%%%%%%%%%%%%%%%%

The proof of \Cref{thm:limit_set} requires the following lemma.
    
\begin{lemma}\label{lemma:limmit_kappa}
        It holds that $n F(\hat{c}_n^*) \rightarrow \hat{\kappa}^*$ as $n \rightarrow \infty$.
\end{lemma}

\begin{proof}[Proof of \Cref{lemma:limmit_kappa}]
        Note that $\Omega(c)$ approaches $-F(c)/f(c)$ for $c> \underline{c}$ as $n$ goes to infinity, while $\Omega(0)=\sum_l w^l q^l \mu^l$ for all $n$. Hence, $\hat{c}^*_n \rightarrow \underline{c}$ as $n \rightarrow \infty$.
        
        From here, we can apply the same arguments as in the proof of Proposition 9 in \cite{gersbach2023}  by replacing $V \frac{1-(1-qF(c_n))^n}{c_n}$ with
        $\sum_l v^l \mu^l \frac{1-(1-q^l F(\hat{c}_n^*))^n}{\hat{c}_n^*} + v_a \frac{1-(1-q_a F(\hat{c}_n^*))^n}{\hat{c}_n^*}$ and $V/\underline{c}$ with $\overline{v}/\underline{c}$. In step 3, the argument still holds, due to the fact that the sum of strictly monotone and continuous functions is also strictly monotone and continuous.
    \end{proof}\\
    
\begin{proof}[Proof of \Cref{thm:limit_set}]
For the proof of the theorem, we remove first the budget constraint for the prize for the first bug. Later, the constraint is reintroduced by taking the intersection of $M_n$, respectively $M_\infty$, with a closed subset of $\mathbb{R}^{L+2}$, and therefore, the convergence will not be affected. Take an arbitrary $(\boldv^*_n, v_{a,n}^*, q_{a,n}^*) \in M_n$. Define $(\boldv^*, v_a^*, q_a^*) \in M_{\infty}$ by setting $v_a=v_{a,n}$, $q_a^*=q_{a,n}^*$, $(v^2, \dots, v^L)=(v^2_n, \dots , v^L_n)$ and
    \begin{equation*}
        v^1=\frac{\underline{c}}{\mu^1(1-e^{-q^1\hat{\kappa}^*})}\left(\hat{\kappa}^* \underline{c} - \sum_{l \neq 1} v^l \mu^l \left(1- e^{-q^l \hat{\kappa}^*}\right) - v_a \left (1- e^{-q_a \hat{\kappa}^*} \right) \right).
    \end{equation*}
    Recall that $v^1_n$ of the public program is given by
    \begin{equation*}
        v^1_n= \frac{\hat{c}_n^*}{\mu^1P(\hat{c}_n^*; q_1)}\left(nF(\hat{c}_n^*)\hat{c}_n^* - \sum_{l \neq 1} v^l \mu^l P(\hat{c}_n^*;q^l) - v_a P(\hat{c}_n^*; q_a) \right)
    \end{equation*}
    and thus, by taking the limit for every summand separately and \Cref{lemma:limmit_kappa}, we obtain
    \begin{equation*}
        \lim_{n \rightarrow \infty} v^1-v_n^1 = \frac{\underline{c}}{\mu^1(1-e^{-q^1 \hat{\kappa}^*})}\lim_{n \rightarrow \infty} v_a \left( (1-e^{-q_a \hat{\kappa}^*})- P(\hat{c}_n^*;q_a) \right).
    \end{equation*}

    Since the sequence $e^{-nx}-(1-x)^n$ converges uniformly on bounded intervals and $q_a \in [0,1]$, the right-hand side converges uniformly, and therefore
    \begin{equation*}
        \sup_{m_m \in M_n} \inf_{m \in M_\infty} \|m_n -m\|_2 \longrightarrow 0, \quad \text{ as } n \rightarrow \infty. 
    \end{equation*}
    In the same way, one can prove that
    \begin{equation*}
        \sup_{m \in M_\infty} \inf_{m_n \in M_n} \|m_n -m\|_2 \longrightarrow 0, \quad \text{ as } n \rightarrow \infty, 
    \end{equation*}
    and therefore, $d(M_n, M_\infty) \rightarrow 0$ in the Hausdorff distance.
\end{proof}

%%%%%%%%%%%%%%%%%%%%%%%%%%%%%%%%%%%%%%%%%%%%%%%%%%%%%%%%%%%%%%%%%%%%%%
%%%%%%%%%%%%%%%%%%%%%%%%%%%%%%%%%%%%%%%%%%%%%%%%%%%%%%%%%%%%%%%%%%%%%%

\begin{proof}[Proof of \Cref{thm:public_artificial_bug}]
    The proof of the first equivalence relation is analogous to the proof of \Cref{thm:artificialbug}. For the second equivalence, note that $\kappa^l(\overline{v})=\lim_{n \rightarrow \infty} n F(c_n^l(\overline{v}))$ for all $l \in \{1, \dots, L\}$  by \Cref{lemma:limits}. Assume that $\tilde{\kappa} > \kappa_0(\overline{v})$, which is equivalent to
    \begin{equation*}
        \lim_{n \rightarrow \infty} \min\{n F(c_n), n F(c_a(\overline{v}))\} > \lim_{n \rightarrow \infty} \max_l \{n F(c_n^l(\overline{v}))\}.
    \end{equation*}
    This is the case exactly when there exists $N$, such that for all $n \geq N$, we have 
    \begin{equation*}
    \min\{n F(c_n), n F(c_a(\overline{v}))\} > \max_l n F(c_n^l(\overline{v})).
    \end{equation*} 
    But is equivalent to $c_{0,n}(\overline{v}) > \tilde{c}_n$ for $n \geq N$, since $F(c)$ is strictly increasing in $c$ as a cumulative distribution function.
\end{proof}

%%%%%%%%%%%%%%%%%%%%%%%%%%%%%%%%%%%%%%%%%%%%%%%%%%%%%%%%%%%%%%%%%%%%%%
%%%%%%%%%%%%%%%%%%%%%%%%%%%%%%%%%%%%%%%%%%%%%%%%%%%%%%%%%%%%%%%%%%%%%%
%%%%%%%%%%%%%%%%%%%%%%%%%%%%%%%%%%%%%%%%%%%%%%%%%%%%%%%%%%%%%%%%%%%%%%
%%%%%%%%%%%%%%%%%%%%%%%%%%%%%%%%%%%%%%%%%%%%%%%%%%%%%%%%%%%%%%%%%%%%%%

\end{document}